\documentclass[sigconf]{acmart}
\usepackage{threeparttable}
\usepackage{bm} 
\usepackage{multirow}
\usepackage{array}
\usepackage[linesnumbered,ruled,vlined]{algorithm2e}
\usepackage{xcolor}

\newtheorem{lemma}{Lemma}
\newtheorem{theorem}{Theorem}
\usepackage{algorithmic}
\usepackage{amsfonts}
\usepackage{svg}
\usepackage{amsmath}
\usepackage{tabularx}
\usepackage{booktabs}
\usepackage{subfigure}
\usepackage[many]{tcolorbox}
\usepackage{graphicx}

\AtBeginDocument{%
  \providecommand\BibTeX{{%
    \normalfont B\kern-0.5em{\scshape i\kern-0.25em b}\kern-0.8em\TeX}}}

\setcopyright{acmcopyright}
\copyrightyear{2023}
\acmYear{2023}
\acmDOI{XXXXXXX.XXXXXXX}

\acmConference[Conference acronym 'XX]{Make sure to enter the correct
  conference title from your rights confirmation emai}{June 03--05,
  2023}{Woodstock, NY}
\acmPrice{15.00}
\acmISBN{978-1-4503-XXXX-X/18/06}

\begin{document}

\title{Massively Parallel Single-Source SimRanks in $o(\log n)$ Rounds}

\author{Siqiang Luo}
\authornote{The authors are ordered alphabetically.} 
\affiliation{%
  \institution{Nanyang Technological University}
  \country{Singapore}
}
\email{siqiang.luo@ntu.edu.sg}

\author{Zulun Zhu\footnotemark[1]}
\affiliation{%
  \institution{Nanyang Technological University}
  \country{Singapore}
}
\email{ZULUN001@e.ntu.edu.sg}



\begin{abstract}
SimRank is one of the most fundamental measures that evaluate the structural similarity between two nodes in a graph and has been applied in a plethora of data management tasks. These tasks often involve single-source SimRank computation that evaluates the SimRank values between a source node $s$ and all other nodes. Due to its high computation complexity, single-source SimRank computation for large graphs is notoriously challenging, and hence recent studies resort to distributed processing. To our surprise, although SimRank has been widely adopted for two decades, theoretical aspects of distributed SimRanks with provable results have rarely been studied.

In this paper, we conduct a theoretical study on single-source SimRank computation in the Massive Parallel Computation (MPC) model, which is the standard theoretical framework modeling distributed systems such as MapReduce, Hadoop, or Spark. Existing distributed SimRank algorithms enforce either $\Omega(\log n)$ communication round complexity or $\Omega(n)$ machine space for a graph of $n$ nodes. We overcome this barrier. 
Particularly, given a graph of $n$ nodes, for any query node $v$ and constant error $\epsilon>\frac{3}{n}$, we show that using $O(\log^2 \log n)$ rounds of communication among machines is almost enough to compute single-source SimRank values with at most $\epsilon$ absolute errors, while each machine only needs a space sub-linear to $n$. To the best of our knowledge, this is the first single-source SimRank algorithm in MPC that can overcome the $\Theta(\log n)$ round complexity barrier with provable result accuracy. 
\end{abstract}





\maketitle

\section{Introduction}
\label{sec:intro}
Evaluating the structural similarity between two nodes in a graph is fundamental in plenty of data management and data mining tasks. Examples include recommendation systems~\cite{nguyen2015evaluation,mo2021agenda}, avoiding customer churn~\cite{luo2019efficient}, spam detection~\cite{benczur2006link}, link prediction~\cite{liben2007link} and graph mining~\cite{chen2020scalable,DBLP:conf/kdd/JinLH11,DBLP:journals/pvldb/LiaoMLLY22,bojchevski2020scaling}. Among many similarity measures, SimRank~\cite{DBLP:conf/kdd/JehW02} is one of the most widely adopted measures over graphs. SimRank is defined based on the intuition that two nodes are similar only when their neighboring nodes are similar. Formally, it uses the following recursive equation to compute the SimRank between two nodes $u$ and $v$, where $c$ is a positive constant factor and $I(u)$ denotes the in-neighbor set of node $u$.
\begin{equation}\label{basic equation}
\small
  s(u, v) =
  \begin{cases}
   1, & u = v \\
   {\frac{c}{|\mathcal{I}(u)||\mathcal{I}(v)|}} \sum_{u'\in \mathcal{I}(u)}\sum_{v' \in \mathcal{I}(v)} s(u', v'), & u\neq v
  \end{cases}
\end{equation}

\noindent
Since it was proposed by
Jeh and Widom~\cite{DBLP:conf/kdd/JehW02}, SimRank gains increasing popularity in various application domains, e.g., social analysis~\cite{DBLP:journals/pvldb/ZhengZF0Z13}, collaborative filtering~\cite{DBLP:journals/pvldb/AntonellisGC08}, and nearest neighbor
search~\cite{DBLP:conf/icde/LeeLY12}.

\subsection{Distributed SimRank Computation} 
We focus on {\it single-source SimRank computation}, whose importance has been uncovered in a plethora of recent studies~\cite{DBLP:journals/pvldb/0012XFC00M20, DBLP:journals/pvldb/LiFLCCL15,wang2020exact,shi2020realtime, wang2021exactsim, DBLP:conf/kdd/MaeharaKK14,tian2016sling,kusumoto2014scalable}. Given a graph $G$, a single-source SimRank for source node $s$ evaluates the SimRank values between $s$ and all other nodes in the graph. Single-source SimRank computation is widely used in applications where a ranking of the other objects with respect to an object is required. For example, it can be applied in search engines to locate the most similar web pages to a given one \cite{fogaras2005scaling}, or in social network services to recommend new friends to a given user \cite{DBLP:conf/kdd/HeFLC10,nguyen2015evaluation}, or act as a ranking measurement to cluster objects \cite{cai2008s}. 

Given a graph of $n$ nodes, computing single-source SimRank is challenging for large graphs because it inherently involves $O(n)$ times of pairwise SimRank evaluations, each of which can already be too costly. Particularly, following the recursive form in Equation~\ref{basic equation}, computing $s(u,v)$ requires accessing many pairs of nodes in the graph, leading to $O(n^2)$ complexity. To address the efficiency issue, recent works~\cite{jiang2017reads,zhang2017experimental,shao2015efficient} employ a random-walk-based approach to approximate the SimRank values. The main idea is to translate the SimRank computation into estimating the meeting probability of two decay-based random walks from the two source nodes (See Section~\ref{sec:pre random} for more details). The computational complexity is then dependent on the number of random walks sampled, achieving a significant speed-up. However, even with the improved approach, it is still challenging to compute single-source SimRanks when $n$ is large.

Therefore, it is increasingly popular to apply {\it distributed computation}~\cite{DBLP:journals/pvldb/0012XFC00M20, DBLP:journals/pvldb/LiFLCCL15,DBLP:journals/tkde/SongLGZWY18} to SimRank computation, which involves multiple machines to compute SimRank values in a collaborative manner, ultimately scaling up the computation to large graphs. 
%
%
%
Existing distributed SimRank algorithms~(e.g.,~\cite{DBLP:journals/pvldb/LiFLCCL15,DBLP:journals/tkde/SongLGZWY18,DBLP:journals/pvldb/0012XFC00M20}) mostly focus on empirical evaluation, 
and non-trivial theoretical analysis with provable approximation result guarantees is rarely given. In this paper, we aim to conduct a non-trivial analysis on computing single-source SimRank values in a distributed setting. We focus on approximate SimRank algorithms that output SimRank values with at most $\epsilon$ absolute errors, where $\epsilon\in (0,1)$ is a given error threshold. We study the topic based on a well-known distributed computation model named MPC (Massively Parallel Computation)~\cite{karloff2010model,andoni2014parallel,beame2017communication,goodrich2011sorting}, which has been widely adopted for theoretical analysis for distributed algorithms and applied to analyzing various data mining and graph processing tasks~\cite{DBLP:conf/podc/Behnezhad0DFHKU19,DBLP:journals/talg/CzumajDP21,DBLP:conf/podc/GhaffariGKMR18,DBLP:conf/focs/BehnezhadDELM19,nowicki2021dynamic}. MPC considers a set of machines, each of which can afford a space of $S$ words. The number of machines $M$ is set to be some integer in $\Tilde{\Theta}(F/S)$ to allow the distributed system to hold the input, where $F$ denotes the size of the input and $\Tilde{\Theta}(\cdot)$ hides logarithmic factors compared with $\Theta(\cdot)$~\footnote{This setting follows many existing works, e.g., ~\cite{chang2019complexity,ghaffari2022massively,biswas2021massively,ghaffari2020improved}. We also note that some other studies on MPC, e.g., ~\cite{DBLP:conf/pods/QiaoT21, tao2022parallel,hu2021cover,hu2020massively}, enforce a stricter requirement that $M=\Theta(F/S)$.}. In MPC, each machine holds part of the data and communicates its local results with other machines via a synchronized message-passing communication round, subject to the constraint that each machine per round can send/receive messages of size at most $S$. 

In modeling distributed computation, the local computation cost within each machine is typically dominated by the synchronization cost among the machines. Hence, by convention, the local computation cost is omitted in MPC analysis and the focus is on reducing communication rounds. 
Furthermore, there is an intrinsic trade-off between the number of communication rounds and space-per-machine. Consider the SimRank computation over a graph with $m$ edges and $n$ nodes. $S=\Omega{(m+n)}$ is a trivial case because all the computation can be done locally in a machine, thus $O(1)$ rounds are sufficient. For large graphs, however, it is more important to consider a sub-linear space setting (i.e., $S=o(n)$) which is typical in many distributed systems and existing studies (e.g., ~\cite{DBLP:journals/talg/CzumajDP21, DBLP:conf/stoc/LackiMOS20,DBLP:conf/podc/Behnezhad0DFHKU19}). Under such a setting, communication among the machines is necessary for capturing the whole graph view, and the purpose is to minimize communication rounds for single-source SimRanks approximation.

\vspace{1mm}
\noindent
{\bf Open Problem.} We analyze existing representative distributed SimRank algorithms and summarize their requirements regarding communication rounds and machine space in Table~\ref{tab:complexity}\footnote{For the limitation of space, we leave the related work analysis in Appendix \ref{appen:baseline note}. As the authors did not give an analysis based on the MPC model, we analyze them on our own.}. In a nutshell, they require either $\Omega(n)$ per-machine space (UniWalk \cite{DBLP:journals/tkde/SongLGZWY18}) or $\Omega{(\log n)}$ communication rounds (DISK \cite{DBLP:journals/pvldb/0012XFC00M20} and CloudWalker \cite{DBLP:journals/pvldb/LiFLCCL15}).  
We remark that overcoming the $\Theta(\log n)$ round complexity barrier in MPC for natural problems is usually challenging and it has attracted tremendous interest~\cite{luo2022distributed,DBLP:conf/stoc/LackiMOS20,DBLP:conf/aaai/Luo19,DBLP:conf/focs/BehnezhadDELM19,DBLP:journals/corr/abs-1807-05374,DBLP:journals/talg/CzumajDP21,luo2020improved} to improve the distributed computation down to $O(\textbf{poly}(\log \log n))$ rounds. In Section~\ref{sec:warm-up} we will also highlight why an intuitive idea based on the state-of-the-art SimRank algorithm cannot achieve this goal. Hence, a natural {\it open theoretical problem} is raised:

{\it Given an error threshold $\epsilon\in (0,1)$, let an approximate single-source SimRank algorithm be an algorithm that outputs SimRank values with $\epsilon$ absolute errors. Is there a distributed approximate single-source SimRank algorithm over a graph of $n$ nodes that can be finished in $o(\log n)$ rounds, while each machine only needs $o(n)$ space?}

\begin{table*}[tb]

\centering
\label{Table1}
\caption{Communication rounds of distributed single-source SimRank Computation on a graph of $n$ nodes and $m$ edges. We explore the existing distributed or parallel SimRank algorithms and analyze them in terms of communication rounds and space per-machine. The analysis is based on synchronized round communication. The total space cost per round is the sum of the space occupied by each machine. The space per-machine refers to the worst case space cost in a machine. Due to unbalanced workloads, the worst-case per-machine space complexity can be the same as the total space complexity.
}

\label{tab:complexity}
\setlength{\tabcolsep}{3mm}{
\begin{threeparttable}
\begin{tabular}{c|ccccc} \toprule[1pt]
{\bf Algorithm} & {\bf Rounds} & {\bf Total space cost per round} & {\bf Accuracy error}& {\bf Space cost per-machine} \\ \midrule[0.5pt]
CloudWalker \cite{DBLP:journals/pvldb/LiFLCCL15} \tnote{1} & $O(\frac{nl(I+l)}{b})$  & $O\left(\frac{bl^2\log n}{\epsilon_p^{2}}+n+m\right)$ & no guarantees & $O\left(\frac{bl^2\log n}{\epsilon_p^{2}}+\frac{m+n}{M}\right)$\\
UniWalk \cite{DBLP:journals/tkde/SongLGZWY18} \tnote{2}& $O(l)$&
$O\left(\frac{n^2 l\log n}{\epsilon^2}\right)$ 
& $\epsilon$ &
$O\left(\frac{n^2 l\log n}{\epsilon^2}\right)$ 
\\

DISK \cite{DBLP:journals/pvldb/0012XFC00M20}\tnote{3} & $O\left(\frac{\log^2 n}{ \epsilon_d^{2}}+K\right)$ & $O\left(m \log n+Kn\right)$ &$\frac{c\left(1-c^{K}\right) \epsilon_d}{1-c}+c^{K+1}$& $O\left(\frac{m\log n+Kn}{M}\right)$ \\
{\bf Ours\tnote{4}} & $O\left(\log^2\log n\right)$ & 
$O\left(m\log n+n^{1+o(1)}\log^{4.5} n+\frac{n\log^{3+o(1)}n}{2\left(\epsilon-\frac{3}{n}\right)^2}\right)$
& $\epsilon$ & $O(n^\alpha)$\\
 \bottomrule[1pt]
\end{tabular}
\begin{tablenotes}
       \footnotesize
       \item[1] $\epsilon_p$ is the error when estimating the random walk distribution, $b$ is the parameter that controls the number of nodes to be handled in a single machine, and $I$ is the number of iterations in the Jacobi method. $l$ is a user-defined walk length or number of steps.
       \item[2] $l$ is a user-defined walk length or number of steps. 
       \item[3] $K$ is the number of truncated terms of linearized SimRank and $\epsilon_d$ is an internal estimation error threshold.
       \item[4] $\alpha <1$ and $n^{\alpha}>{\log^5n}/({\log\frac{1}{\sqrt{c}}})$.
\end{tablenotes}
\end{threeparttable}
}

\end{table*}

\subsection{Our Main Results} 
In this paper, we give a positive answer to the aforementioned open problem and present a distributed single-source SimRank algorithm that suits the Massively Parallel Computing (MPC) model~\cite{karloff2010model}. Particularly, for a graph of $n$ nodes and $m$ edges, we focus on the MPC model that involves a set of machines, each having a sub-linear space $S=n^{\alpha}$ for some $\alpha \in (0,1)$. 
Our main results can be stated by the following theorem.

\begin{theorem}\label{the:contri}
{Given a query source node $s$ in a graph of $n$ nodes and $m$ edges, there is an algorithm that computes $\epsilon$-absolute-error guaranteed SimRank values between $s$ and all the other graph nodes using $M$ machines in $O\left(\log^2\log n\right)$ communication rounds with high probability~\footnote{We say an event happens with high probability, if there exists a constant $\tau>0$ such that the event happens with probability at least $1-\frac{1}{n^\tau}$, where $n$ is the number of graph nodes.}. This algorithm only requires that the space per machine is $S=n^{\alpha}$ for some $\alpha<1$, and $M$ is some value of $\Tilde{O}((m+n)/S)$.}
\end{theorem}

To the best of our knowledge, this is the first distributed single-source SimRank algorithm that achieves sub-$\log n$ communication rounds while only requiring strongly sub-linear per-machine space.
\section{Preliminaries}
We first introduce the distributed model we employed. Then we discuss the basic preliminary concepts and algorithms of SimRank computation. The frequently used notations are listed in Table \ref{table:fre_notations}. 

\begin{table}[t!]
    \centering
    \caption{Frequently used notations.}\label{table:fre_notations}
    \vspace{-0.5em}
    \begin{tabular*}{\hsize}{@{}@{\extracolsep{\fill}}|l|l|@{}}
    \hline
    {\bf Notation} & {\bf Description}  \\
    \hline\hline
        $G\left(V, E\right)$& Directed graph $G$ with node set $V$ and edge set $E$ \\  
        $n, m$& $n = |V|, m = |E|$\\
         $s(v, u)$& SimRank score between node $v$ and $u$\\
$\tilde{s}(v, u)$& Estimated SimRank score between node $v$ and $u$\\
$l$&Length or steps of random walks\\
$\epsilon$&Error threshold for SimRank scores \\
$c$& Decay factor in the definition of SimRank scores  \\
$M$& Number of all machines\\
$S$& Available space per machine\\
$\alpha$& Constant factor to measure the space\\
$N_l$& Actual number of Length-$l$ random walks\\
    \hline
    \end{tabular*}
    \label{tab:parameter description}
\end{table}




\subsection{Distributed Computation Model}
Distributed computation is one of the most important techniques in the era of big data to address various applications~\cite{kuhn2010distributed,DBLP:conf/edbt/LuoZXYLK23,kossmann2000state,DBLP:journals/pvldb/0012XFC00M20, DBLP:journals/pvldb/LiFLCCL15,DBLP:journals/tkde/SongLGZWY18,luo2022distributed,luo2014distributed,luo2012disks,klauck2014distributed}. In recent years, Massively Parallel Computation (MPC) \cite{karloff2010model} becomes a popular theoretical framework in modeling the complexity of a distributed algorithm, because it closely simulates the situation of general distributed computation. An MPC model has three important parameters: the input data size $F$, the number of involved machines $M$, and the space capacity (words) $S$ on each machine. 
It is required that when given $F$ and $S$, the number of machines $M$ should be of $\Tilde{O}(\frac{F}{S})$, where {$\tilde{O}(\cdot)$ hides a poly-logarithmic factor}. Our main focus on the MPC model derives from several perspectives as follows:

\vspace{1mm}
\noindent
\textbf{Space.}
Consider an input graph of $n$ nodes. MPC for graph algorithms can be categorized into three types: strongly \textbf{super-linear} space ($S = n^{1+\omega}$) for some constant $\omega>0$, near \textbf{linear} space ($S = \Theta (n)$), and strongly \textbf{sub-linear} space ($S = n^{\alpha}$) for some constant $\alpha \in (0,1)$. A super-linear model can typically be employed with a local algorithm, which loses the generality to be employed on a large scale of data. Hence, many studies focus on sub-linear settings that better capture the scalability of a distributed system. 
In this paper, we focus on the strongly sub-linear setting where each machine cannot even store the whole set of graph nodes. 

\vspace{1mm}
\noindent
\textbf{Communication Rounds.}
The computation in the MPC model is based on {\it communication rounds}. Initially, each edge is randomly assigned to a machine. We assume each node has an integer ID from $1$ to $n$, and each machine has an integer ID from $1$ to $M$. The computation proceeds in synchronized rounds. At the beginning of a round, each machine may receive the messages sent from some other machines in the previous round. During a round, every machine conducts local computation based on its local data or messages received. Then each machine will send the computed results, packed as messages, to target machines. Each machine creates message packages to be routed onto the network and hence the size of messages sent/received per machine in one round shall not exceed the space capacity $S$. A new round starts only after the end of the previous round. The number of rounds needed for program execution is called \textit{round complexity}  \cite{DBLP:conf/focs/GhaffariKU19}, which describes the cost of a distributed computation as the dominating cost often comes from the costly communication among machines. The main principle of designing algorithms in the MPC model is to achieve low round complexity subject to the space constraint on each machine. 
%
%

\vspace{1mm}
\noindent
{\textbf{Existence of an MPC Algorithm.} }Given per-machine space $S=n^{\alpha}$, we say there exists an MPC algorithm, if for any $n\ge n_0$ where $n_0$ is a constant, there exists a distributed algorithm under the space constraint that uses $M$ machines and 
$M=\tilde{O}(\frac{m+n}{S})$. For ease of analyzing the existence of MPC algorithms, we first give a preparation lemma regarding machine space expansion, as follows.
\begin{lemma}\label{lemma:pre}
If there exists an $O(f(n))$-round distributed graph algorithm that works for per-machine space $S={\Theta}(n^{\alpha})$ for any $0<\alpha < 1$ using $M=\tilde{O}(\frac{m+n}{n^{\alpha}})$ machines, and $f(n)$ is a function not related to $\alpha$ (i.e., $\alpha$ only contributes a constant factor to the round complexity), then there exists an MPC algorithm that works for per-machine space $S=n^{\alpha}$ with the same round complexity.
\end{lemma}
This lemma eliminates the obstacle of analyzing round complexity when there is a constant factor expansion of the machine space. For example, consider that we have an algorithm with a certain distributed algorithm such that (1) its round complexity hides $\alpha$ factors; (2) applies to any $\alpha \in (0,1)$ with $M=\tilde{O}(\frac{m+n}{n^{\alpha}})$, and (3) the space per machine is $Cn^{\alpha}$ for some constant $C$. Then Lemma~\ref{lemma:pre} guarantees the existence of an MPC algorithm with $n^{\alpha}$ machine space.

\subsection{Approximate SimRank Computation}\label{sec:pre random} 
Given a directed graph $G = (V, E)$, and let $n=|V|$ and $m=|E|$. 
Following~\cite{wang2020exact,shi2020realtime}, we aim to compute approximate SimRank values with constant errors. 
%
%
In particular, given any source node $s\in G$ and a constant error $\epsilon$, we aim to compute the SimRank values between $s$ and any other node $u\in G$, such that 
$|s(s,u)-\tilde{s}(s,u)|\leq \epsilon$, 
where $s(s,u)$ denotes the true SimRank value and $\tilde{s}(s,u)$ denotes the estimated value.


Calculating the SimRank value iteratively according to Equation \ref{basic equation} may occupy large memory space and incur a high computation cost. Therefore, the state-of-the-art approaches employ the following $\sqrt{c}$-decay walk-based computation, first proposed in~\cite{tian2016sling}.

\begin{lemma}\label{meeting prob}
For any two nodes $u, v\in G$, the SimRank between $u$ and $v$ is equal to the meeting probability of two $\sqrt{c}$-decay walks starting at $u$ and $v$ on $\bar{G}$, where two walks meet if there exists an integer $i\ge 0$, such that the $i$-th step of the two walks visit the same node. 

\end{lemma}

Here, a $\sqrt{c}$-decay walk on $\bar{G}$ from a node $u$ is a traversal on $\bar{G}$ such that at each step of the walk, it has $1-\sqrt{c}$ probability to stop at the current node, and otherwise jumps to the next node that is a uniformly chosen out-neighbor of the current node. Here we define a length-$l$ walk as a path that includes $l+1$ nodes and $l$ edges. We also say a length-$l$ walk has $l$ steps. Particularly, we have $(\sqrt{c})^i\cdot (1-\sqrt{c})$ probability to generate a length-$i$ walk from a given starting node.

\vspace{1mm}
\noindent
{\bf Monte Carlo Method.} By sampling $N$ pairs of $\sqrt{c}$-walks from $u$ and $v$, and if $H$ pairs of the walks meet, then $\frac{H}{N}$ is an estimate of $s(u,v)$. The sampling number $N$ controls the estimation accuracy.

\section{Warm-Up: A $O(\log n)$-Round Algorithm with Pseudo-linear Space}\label{sec:warm-up}
In this section, we present an algorithm that directly adapts the $\sqrt{c}$-walk-based computation for single-source SimRank approximation in the MPC model. We show that by properly setting the parameters, the algorithm can finish the computation in $O(\log n)$ rounds with high probability (i.e., with a probability at least $1-\frac{1}{n}$) while it requires pseudo-linear space-per-machine $S=\tilde{O}(n)$.

The main idea of the algorithm is to adapt the $\sqrt{c}$-walk based approach and apply the Monte Carlo method by setting the sample size $N=\frac{\log 2n^2}{2\epsilon^2}=O(\frac{\log {n}}{\epsilon^2})$, i.e., sampling $\sqrt{c}$-walks from the given source node $u$ and any other node $v$ over the reverse graph $\bar{G}$. By the Monte Carlo method, we need to create $N$ pairs of random walks from $u$ and $v$ and determine how many paired-up walks meet. The probability of the walk-meet, $\frac{H}{N}$, estimates $s(u,v)$. Consider a random variable $X_i=1$ if the two walks in the $i$-th pair meet, otherwise $X_i=0$. By Hoeffding's 
Inequality~\cite{hoeffding1994probability} we have the following inequality:

\begin{align}
    {\bm \Pr}{\left[\left|s(u,v)-\frac{H}{N}\right|\ge \epsilon\right]}&=  {\bm \Pr}{\left[\left|\sum_{1\leq i\leq N}\mathbb{E}(X_i)-\sum_{1\leq i\leq N}X_i\right|\ge  \epsilon\right]}  \nonumber \\ &\leq  2e^{-2N\epsilon^2} \text{ (By Hoeffding's Inequality)} \nonumber \\ & ={1}/{n^2}
\end{align}

Since there are $n$ SimRank estimations, by union bound, all the estimated SimRank has an error at most $\epsilon$ with a probability of at least $1-\frac{1}{n^2}\cdot n = 1-\frac{1}{n}$.

Next, we show that sampling $N$ random walks respectively from each node can be performed in parallel, and concurrently walking one step for all walks can be finished in $O(1)$ rounds with machine space $S=\tilde{O}(n)$. Particularly, we assign each edge $(u,v)$ with an ID $u\cdot n +v$. In MPC we can use $O(1)$ communication rounds to sort~\cite{MPA,DBLP:conf/isaac/GoodrichSZ11} the edges based on their IDs and store them from Machine 1 to Machine $M$ in order. To separate the neighbor set of consecutive nodes, we insert a special edge placeholder $(v,0)$ with ID $v\cdot n$ for every node $v$ and sort them together with the actual edges. After sorting, the placeholder $(v,0)$ will be placed right before the first edge starting with $v$. By sorting, each edge also knows its rank in the ordering. Then node $v$ has $r_{v+1}-r_v-1$ neighbors, where $r_v$ (resp. $r_{v+1}$) denotes the rank of $(v,0)$ (resp. $v+1$). Evaluating $\{r_{v+1}-r_v-1|v\in V\}$ can be done in $2$ rounds: suppose each machine has $S=\tilde{O}(n)$, then all $\{r_v|v\in V\}$ can be sent to the same machine to compute {$\{r_{v+1}-r_v-1|v\in V\}$} and send $r_{v+1}-r_v-1$ back to the machine holding edge $(v,0)$. This will not break the $\tilde{O}(n)$ space as the size of the set $\{r_{v+1}-r_v-1|v\in V\}$ is $O(n)$. Furthermore, to store the random walks, when the walk path $P$ ending at node $u$ extends from node $u$ to $v$ along an edge $(u,v)$, the path $P\cup (u,v)$ will be sent to the machine holding $(v,0)$. 

Since the random walks are $\sqrt{c}$-decay, there is at most probability ${\epsilon^3}/{n^3}$ that a random walk has a length at least $3\log_{1/\sqrt{c}}{\frac{n}{\epsilon}}$. Hence all the $Nn$ random walks have length $O(\log \frac{n}{\epsilon})$ with high probability. Noting that it takes $O(l)$ rounds to compute length-$l$ random walks, the round complexity of computing all the random walks is $O(\log \frac{n}{\epsilon})$ with high probability. Furthermore, as there are $Nn=O(\frac{n\log n}{\epsilon^2})$ random walks and each walk has a length $O(\log \frac{n}{\epsilon})$ with high probability, the messages sent among the machines in a round is at most $O(\log {n}\cdot Nn)=O(\frac{n\log n\log \frac{n}{\epsilon}}{\epsilon^2}) =  \tilde{O}({n})$ with high probability (given $\epsilon$ is a constant).

Finally, to compute $\{s(u,v)| v\in V\}$ it remains to check in parallel how many pairs of walk meet, where each pair contains one walk from $u$ and the other walk from $v$. As we will explain shortly in Section~\ref{subsec:Parallel Random Walks}, this step can also be finished in $O(1)$ rounds. 

In summary, the algorithm we presented can finish approximating SimRank computation in $O(\log \frac{n}{\epsilon})$ rounds with high probability over MPC with $S=O(\frac{n\log n\log \frac{n}{\epsilon}}{\epsilon^2})$. One can see that directly adapting the existing random-walk-based algorithms is hard to break the $O(\log n)$-round barrier because in such designs each random walk step costs $O(1)$ rounds and a random walk has $O(\log n)$ steps with high probability. Meanwhile, it is also challenging to further reduce the space-per-machine because there can be hub nodes that are passed through by many random walks sourced at different nodes.

\section{$O(\log^2\log n)$-Round Algorithm with Sub-linear Space}\label{sec:overall}

In this section, we show several important improvements based on the algorithm introduced in the previous section, ultimately reducing the round complexity to $O(\log^2\log n)$ and space per machine to strongly sub-linear. 

To better understand our design, we outline challenges of direct adaptation of $\sqrt{c}$-walk based approach. First, the maximum length of a $\sqrt{c}$-decay walk can be infinite; a straightforward method to compute length-$l$ in a distributed environment easily entails $l$ communication rounds because one step may require one round of communication when some neighbors are located in a different machine. In Section~\ref{sec:warm-up}, we presented an analysis showing that with high probability that the communication rounds can be bounded by $O(\log \frac{n}{\epsilon})$. One downside of such techniques is that the number of rounds is not bounded deterministically. 
Second, evaluating single-source SimRank values for node $s$ requires running random walks from both $s$ and all the other nodes, resulting in a large number of random walks being conducted. This can lead to a high space cost per machine because some hub nodes are prone to be passed by many random walks, and this is the main reason why the algorithm in Section~\ref{sec:warm-up} cannot achieve a strongly sub-linear per-machine space. Therefore, a careful design of the MPC algorithm is required to guarantee a small number of rounds and a low space cost for each machine.
Third, the random walks, once computed in the MPC model, are stored in different machines. As such, detecting whether two walks meet may overload the machine regarding space cost. 
\subsection{New Interpretation of Walk-based Approach}
{To address the aforementioned challenges, we {reinterpret} the $\sqrt{c}$-walk based SimRank computation between node $u$ and node $v$ in a batch manner using the following three operations, which are more MPC-friendly:}

\vspace{1mm}
\noindent
{\bf (a) Random walk generation.} Instead of generating $\sqrt{c}$-decay random walks whose lengths are non-deterministic, we generate random walks with deterministic length distribution. Particularly, to generate $N$ $\sqrt{c}$-walks from $u$, we will generate $N\cdot \sqrt{c}^i(1-\sqrt{c})$ length-$i$ walks for $i\ge 0$, following the corresponding geometric distribution of walk lengths. 
Similar operations are conducted for random walks starting from $v$. Fixing the random walk lengths is more MPC-friendly. As we will show in Section~\ref{subsec:Parallel Random Walks}, we can generate sufficient such random walks both from the source node $s$ and each node $v \in V$ in only $O\left(\log^2\log n\right)$ MPC rounds. 

\vspace{1mm}
\noindent
{\bf (b) Random walk shuffling and matching.} Since the random walks from $s$ (or $v$) are generated in the order of increasing lengths (due to Step (a)), pairing up the random walks from $s$ and $v$ directly for meeting detection is not valid for estimating SimRanks as the $i$-th walks from $s$ and $v$ are correlated (they have the same length). In order to maintain the randomness, we shuffle the random walks computed in Step (a). 

\vspace{1mm}
\noindent
{\bf (c) SimRank computation.} Last, we compute the SimRanks by detecting how many pairs of walks meet and estimating the SimRank value $s(s,v)$ by Lemma~\ref{meeting prob}. 




\subsection{Overview of Main Algorithmic Steps}\label{sec:routing}

Based on the new interpretation, we give our main algorithmic steps in~Algorithm \ref{alg:overall}, which consists of the following five main stages. For ease of presentation, we first outline the main idea in this section and defer the detailed MPC-related operations in Section~\ref{sec:MPC_analysis}.

\setlength{\textfloatsep}{1pt}
\begin{algorithm}[t]
\small
\SetKwInOut{Input}{Input}\SetKwInOut{Output}{Output}
\Input{Graph $G =(V, E)$; source node $s$; decay factor $c$; Montel-Carlo failure probability $\delta$; accuracy error $\epsilon$}
\Output{SimRank scores $\tilde{s}(s, u)$ for each $u \in V$}
\BlankLine
\emph{{Store $G$ across multiple machines}}\;

\emph{{Reverse the edges of $G$ within each machine to form $\bar{G}$}}\;

\emph{$l \gets \log_{\frac{1}{\sqrt{c}}} n$}\;

    \For{$i = 0$ \KwTo $l$ in parallel}{
    $N_i = \left\lceil \frac{\log{2n}}{2\left(\epsilon-\frac{3}{n}\right)^2}\cdot (\sqrt{c})^{i}\cdot (1-\sqrt{c})\right\rceil$\;
    \ForEach{node $v \in V$ in parallel}{
     Generate $N_i$ random walks from node $v$ in $O\left(\log^2\log n\right)$ rounds\; 
    }
    
}
Shuffle and decompose random walks (See Sections~\ref{subsec:Shuffle Random  Walks} and~\ref{subsec:Decompose Random  Walks})\;
    Run \textbf{Algorithm~\ref{alg:meeting}} to calculate $\tilde{s}(s, u)$ for each $u\in V$\; 
\caption{Overall algorithm}
\label{alg:overall}

\end{algorithm}
\setlength{\textfloatsep}{1pt}

\noindent
\textbf{Initial State.} The SimRank algorithms are based on the topology information of the input graph. 
Initially, the whole graph should be partitioned across different machines, as assumed by the MPC model. We assume that each machine holds a random partition of edges of the graph (Line 1 in Algorithm \ref{alg:overall}), and hence each machine holds roughly $\frac{m}{M}$ edges. We also reverse the edges so that the graph represents $\bar{G}$ (Line 2 of Algorithm~\ref{alg:overall}). 




\vspace{1mm}
\noindent
\textbf{Parallel Random Walks Generation.} 
This stage corresponds to Operation (a) mentioned earlier. As we need to compute single-source SimRank values from any source node $s$ to all other nodes, we generate $N$ $\sqrt{c}$-walks from each node. 
Further, we truncate them at $\log_{{1}/{\sqrt{c}}} n$ length (Line 3 of Algorithm~\ref{alg:overall}). By setting $N=\frac{\log{2n}}{2(\epsilon-3/n)^2}$, we show that the truncated random walks will only cause negligible influence on the final accuracy (see Section \ref{sec:MPC_analysis}), and we can guarantee the $\epsilon$ error bound of SimRank values. Another significant problem is when generating $N_i=N\cdot (\sqrt{c})^i (1-\sqrt{c})$ length-$i$ walks, $N_i$ may not be an integer. In Section~\ref{sec:correctness} we give a rounding technique and prove that the rounding still guarantees the error bound. We also prove that these random walks can be obtained in the MPC model using $O(\log^2 \log n)$ communication rounds (Lines 4-7 of Algorithm~\ref{alg:overall} and details in Section~\ref{subsec:Parallel Random Walks}). 


\vspace{1mm}
\noindent
\textbf{Shuffling Random Walks.}
To guarantee the randomness of each pair of random walks in our Monte Carlo simulation, the generated random walks have to be shuffled (Line 8 of Algorithm \ref{alg:overall}). The shuffling operation in the MPC model is not as trivial as in a single machine because the generated walks from the same node can be located in different machines. Shuffling these walks incurs communication between machines, which may violate the space-bound in each machine. We discuss this issue in detail in Section~\ref{sec:MPC_analysis}. 

\vspace{1mm}
\noindent
  \textbf{Decomposing Random Walks.}
  After shuffling, we pair the $j$-th walk from the source node $s$ with the $j$-th walks from other nodes respectively, and detect whether two walks in each pair meet. Detecting whether two walks meet is relatively simple in the single-machine setting because all information can be loaded locally. 
  However, in a distributed environment, we need to detect, many times, whether two paired-up walks meet while the walks are stored at different machines. This poses drastically different challenges. To address the challenges, we decompose the generated random walks into tuples, each containing one visited node in the walk (except the starting node). Each tuple contains five elements to convey all the information needed to detect whether the walk intersects another walk in the specified node. As shown in Figure~\ref{fig:meeting}, there is a walk $\{s, v_1, v_2, v_3\}$ decomposed into three tuples $(j, 1, v_1, 3, s)$, $(j, 2, v_2, 3, s)$, $(j, 3, v_3, 3, s)$, where $j$ implies that the tuple is decomposed from the $j$-th walk of those starting from $s$. The remaining elements in the tuple are the {\it walk step}, {\it visited node at the particular walk step}, {\it the total steps in the walk}, and {\it source node}, respectively. 
  

\vspace{1mm}
\noindent
\textbf{Computing SimRank Values.}
Evaluating the SimRank value between $s$ and $u$ is based on the probability that two walks from $s$ and $u$ meet.
This can be detected by sorting the decomposed tuples (Line 9 of Algorithm~\ref{alg:overall}$\to$ Line 1 of Algorithm~\ref{alg:meeting}). Figure~\ref{fig:meeting} shows the $j$-th walks from $s$ and $u$ respectively. The two walks meet because they share the same node $v_2$ in their third walk steps. Corresponding to $v_2$, there are two decomposed tuples from the two walks, denoted by $(j, 2, v_2, 3, s)$ and $(j, 2, v_2, 4, u)$. We note that they {\it share the first three elements}. In general, it is easy to verify that if two walks meet, there must be two decomposed tuples, each from one walk, that share the first three elements in the tuple. More formally, suppose there are two $j$-th walks $\{v_0, v_1, \ldots, v_{l_1}\}$ and $\{u_0, u_1, \ldots, u_{l_2}\} (0\le l_1, l_2\le \log_{{1}/{\sqrt{c}}} n)$ 
starting from $v_0$ and $u_0$. Walk $\{v_0, v_1, \ldots, v_{l_1}\}$ is decomposed into $l_1$ tuples, denoted by $(j, i, v_i, l_1, v_0)$ ($1\leq i \leq l_1$); Walk $\{u_0, u_1, \ldots, u_{l_2}\}$ is decomposed into $l_2$ tuples, denoted by $(j, k, u_k, l_2, u_0)$ ($1\leq k\leq l_2$). Then, if the two walks meet, or equivalently, they share the same node $v_i$ in the same step $i$, then there must be $v_i=u_i$ and the corresponding two tuples (for $v_i$ and $u_i$) share the first three elements. Again, the challenges are in a proper rearrangement of the tuples across the machines, and the collection and aggregation of the walk-meeting cases within each machine. We leave MPC details in Section~\ref{sec:MPC_analysis}.  


\begin{figure}

\centering
\includegraphics[width=2.5in]{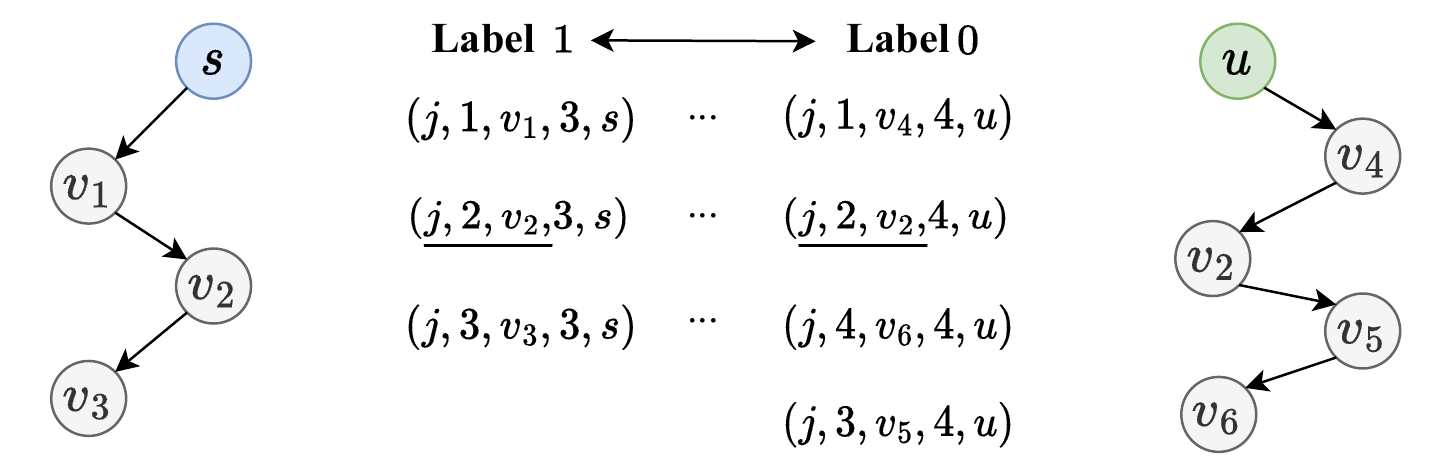}
\caption{\label{fig:meeting}Detect meeting using tuple association.}

\end{figure}



With the above five stages, we highlight the major differences between our approach and the classic $\sqrt{c}$-walk-based algorithms: \textbf{(a)} the lengths of the walks are regularized to geometric distributions, and longer walks are truncated and we show that this would not violate the error tolerance in Section~\ref{sec:correctness}; \textbf{(b)} most steps are redesigned non-trivially so that all the steps can be efficiently implemented in the MPC model (see details in Section~\ref{sec:MPC_analysis}).
\section{Detailed MPC Operations}\label{sec:MPC_analysis}

We present the algorithms appearing in Section \ref{sec:routing} in more details from the perspective of MPC model based implementation. 

\subsection{Parallel Random Walks
Generation}\label{subsec:Parallel Random Walks} 

Performing random walks in the MPC model will incur communication rounds because the next sampled neighbor for each node can be located in different machines. Let us consider the situation of a node $v_0$ sampling its neighboring node $v_1$ in the walk. In the MPC model, as $v_0$ and $v_1$ can be in different machines, forming a walk of $v_0\to v_1$ may need one communication round. In general, forming an $l$-step random walk in the MPC model typically incurs $l$ communication rounds. 
%
For the purpose of reducing communication rounds, we firts introduce a result provided by WRME \cite{DBLP:conf/stoc/LackiMOS20}, and then give our extension in Theorem~\ref{thm:walk}:

\begin{theorem}\label{stoc}
Let $G$ be a directed graph. Let $N$ and $l$ be positive integers such that $l=o(S) / \log ^{3} n$, where $S$ is the available space per machine. For the task that samples $N$ independent random walks of length $l$ starting from each node $v$ in $G $, there exists an MPC algorithm that runs in $O\left(\log ^{2} \log n+\log ^{2} l\right)$ rounds and uses $O\left(m+n^{1+o(1)} l^{3.5}+N n l^{2+o(1)}\right)$ total space and strongly sub-linear space per machine $S=n^{\alpha}$ ($0<\alpha<1$). The algorithm is an imperfect sampler that does not fail with probability $1-O\left(n^{-1}\right)$.
\end{theorem}

Essentially, Theorem~\ref{stoc} states that the generation of a length-$l$ random walk from {\it every} node can be round-efficient in the MPC model. Particularly, it takes only $O\left(\log ^{2} \log n+\log ^{2} l\right)$ communication rounds to generate $N$ length-$l$ random walks, each with different source node. However, even with Theorem~\ref{stoc}, we cannot directly give a reasonable round-complexity for SimRank evaluations if we use the original $\sqrt{c}$-decay walk-based method. The reason is that the length $l$ of a $\sqrt{c}$-decay walk can be infinite. To address this issue, as we show in Algorithm~\ref{alg:overall} Line 3 and Line 5, we carefully design the truncated walk length (Line 3) and the number of random walks to be sampled for each walk length (Line 5). The truncated length guarantees a reasonable bound of $O\left(\log ^{2} \log n\right)$ communication rounds when $l=\log_{{1}/{\sqrt{c}}} n$, and the number of random-walk samples ensures the SimRank estimation accuracy (we will formally prove it in Section~\ref{sec:correctness}).

\begin{figure}

\centering
\includegraphics[width=3.2in]{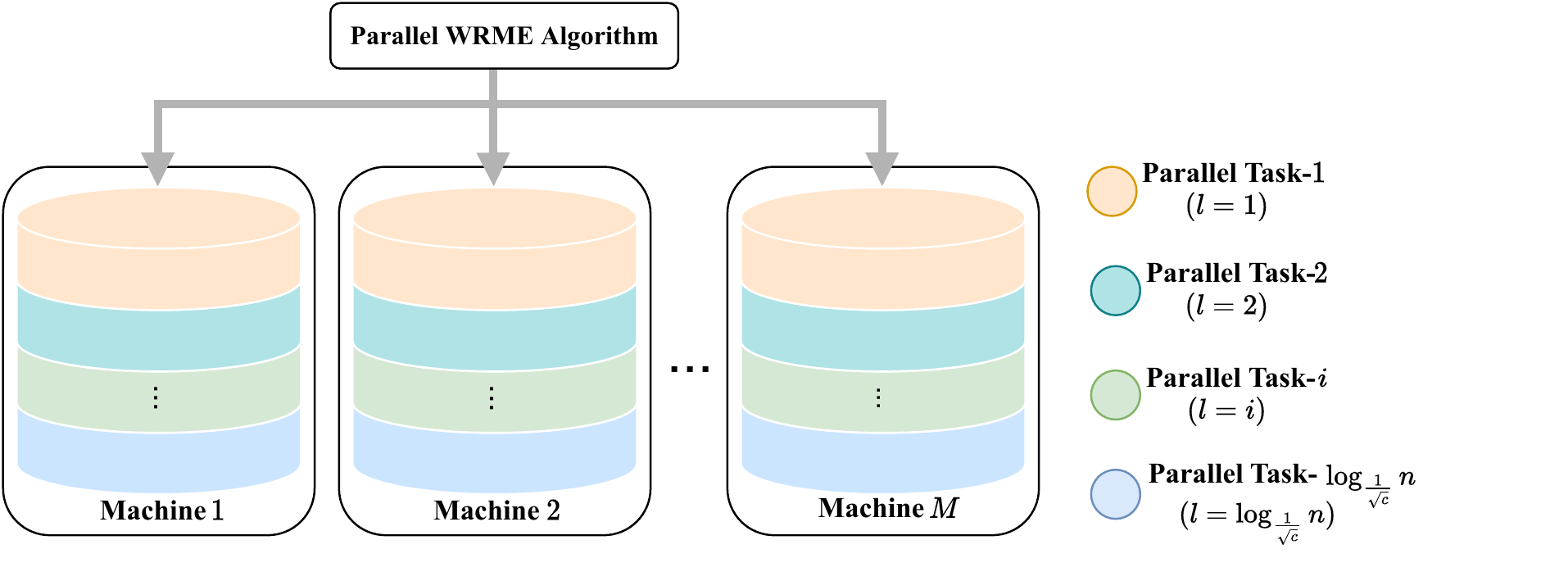}
  \caption{\label{fig:random walk} Each machine divides its space to handle multiple tasks in parallel.  }

\end{figure}

\vspace{1mm}
\noindent
{\bf Parallel Random Walks in $O(\log^2\log n)$ Rounds.} To effectively apply Theorem~\ref{stoc}, we let Task-$i$ be the generation of $N_i$ length-$i$ walks from every node. We let all Task-$i$ ($1\leq i\leq \log_{{1}/{\sqrt{c}}} n$) be conducted in MPC in parallel. 
We apply $\log_{{1}/{\sqrt{c}}} n$ WRME algorithm~\cite{DBLP:conf/stoc/LackiMOS20} instances {\it concurrently} in the MPC model, where each instance corresponding to one task. 

Unfortunately, concurrently conducting multiple WRME algorithms will incur a space cost higher than $n^{\alpha}$ in each machine because by default each algorithm instance can incur a local cost up to $n^{\alpha}$. To address this issue, we let $\beta$ satisfy that $n^{\beta}=n^{\alpha}/ \log_{{1}/{\sqrt{c}}}n$, and we apply the WRME algorithm concurrently for each Task-$i$ with $S=n^{\beta}$. Particularly, each machine space $n^{\alpha}$ is divided into subspaces of size $n^{\beta}$ to compute Task-$i$, as shown in Figure~\ref{fig:random walk}. By careful analysis of the round complexity and total space cost, we give the following theorem. 
\begin{theorem}\label{thm:walk}
The random walk generation in Algorithm~\ref{alg:overall} can be done in $O\left(\log^2\log n\right)$ rounds with high probability, with machine space $S=n^{\alpha}$ such that $n^{\alpha}>{\log^5 n}/({\log\frac{1}{\sqrt{c}}})$ and a total space $O\left(m\log n+n^{1+o(1)}\log^{4.5} n+\frac{n\log^{3+o(1)}n}{2(\epsilon-\frac{3}{n})^2}\right)$.
\end{theorem}
We note that the WRME algorithm is an imperfect sampler that fails with at most a probability $O(n^{-1})$. We only apply $O(\log n)$ times the sampler, which can still easily guarantee that our algorithm is successful with high probability (because $O(\log n) = o(n^{\tau})$ for any constant $\tau>0$).

\subsection{Shuffling Random  Walks} \label{subsec:Shuffle Random  Walks}

We need to shuffle the generated $N$ random walks sourced at each node to fully mimic the behavior of generating a $\sqrt{c}$-walk. {The reason for shuffling is that when we select two random walks that start from two nodes and calculate their meeting probability, the lengths of these two walks are not bound to be the same. This shuffling operation helps us
 guarantee the randomness in our Monte Carlo simulation of the generated $\sqrt{c}$-decay walks.}
 
 \vspace{1mm}
 \noindent
 {\bf Shuffling in $O(1)$ Rounds.}
 We note that sorting in MPC can be done in $O(1)$ rounds~\cite{DBLP:conf/isaac/GoodrichSZ11} as long as the number of items is $O(n^{\eta})$ for some constant $\eta>0$. Observe that $N=O(\frac{\log n}{\epsilon^2})$, and hence there exists $\eta$ such that $nN=O(n^{\eta})$. Hence, sorting the $nN$ walks all together can be finished in constant rounds in MPC, where the comparison of two walks is based on the comparison of the corresponding nodes in the walks. As such, all the walks sourced at the same node will be clustered.
 
 For sufficiently large $n$ we have $n^{\alpha}>2N\log_{1/\sqrt{c}}n$, and the walks sourced at the same node, occupying a space of $N\log_{1/\sqrt{c}}n$, can be held in a single machine. Since all the walks are sorted and $\frac{n^{\alpha}}{N\log_{1/\sqrt{c}}n}$ may not be an integer, each set of $N$ walks sourced at the same node may go across two machines. Let ${\bm W}_{u}$ be the set of walks sourced at node $u$. We discuss two cases to shuffle the walks of the same source node.
 
 {\bf Case 1:} If the walks in ${\bm W}_u$ are fully located in a machine, shuffling is done locally. We shuffle all such sets of walks, and for each set we assign numbers from 1 to $N$ to each walk in ${\bm W}_{u}$ after shuffling. 
 
 {\bf Case 2: }If the walks sourced at the same node are located at two consecutive machines. Without loss of generality, we denote the source node as $u$, and there are $N_0$ walks in the first machine and $N-N_0$ in the second. We then swap those $N-N_0$ walks in the second machine with the last $N-N_0$ walks sourced at node $u-1$, as illustrated in Figure~\ref{fig:shuffling}. Since $n^{\alpha}>2N\log_{1/\sqrt{c}}n$ for sufficiently large $n$, it is guaranteed that those $N-N_0$ walks sourced at $u-1$ are located in the first machine and have been shuffled locally in Case 1. Once swapped, the shuffling within ${\bm W}_u$ can be done locally. We note that by guaranteeing $n^{\alpha}>2N\log_{1/\sqrt{c}}n$ we can make sure the swaps are not conflicting and thus can be done in $1$ round.

Any set ${\bm W}_u$ can be shuffled in Case 1 or Case 2. During shuffling, each walk can be labeled by 1 to $N$ sequentially, and the $j$-th walk refers to the walk assigned with a number $j$.
\begin{figure}

\centering
\includegraphics[width=2.6in]{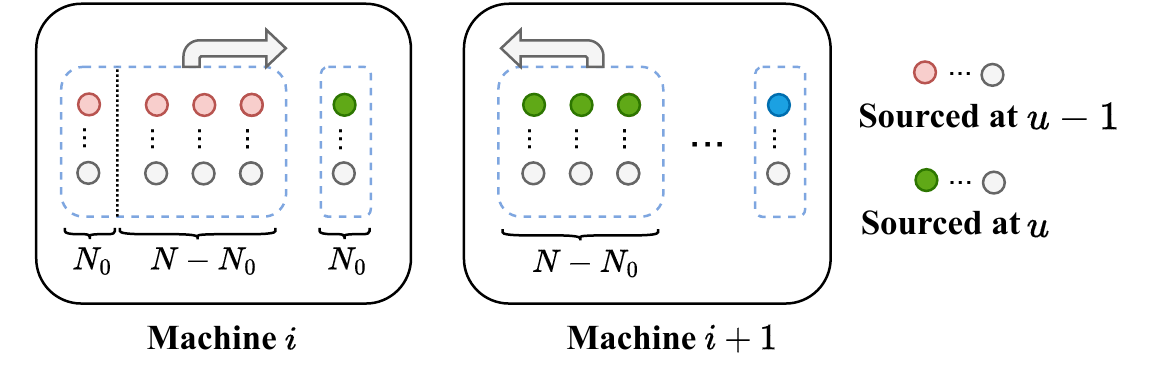}
  \caption{\label{fig:shuffling} Swapping $N-N_0$ random walks.}

\end{figure}

\subsection{Decomposing Random  Walks} \label{subsec:Decompose Random  Walks}

\noindent
{\bf Decomposition in $O(1)$ Rounds.} 
Recall that each length-$t$ walk $\{v_0, v_1, \ldots, v_{t}\}$ is to be decomposed into $t$ tuples each of which has 5 elements: $(j, i, v_i, t, v_0)$, for $1\leq i\leq t$. Tuple $(j, i, v_i, t, v_0)$ implies that in the $j$-th random walk of the walks starting from $v_0$, the $i$-th step of the walk visits node $v_i$ and the walk has overall $t$ steps.
In the MPC model, the decomposition of a random walk can be done within the machine holding the walk. While each walk will generate multiple tuples, the total space is only amplified by a constant factor (i.e., $5$ times) because each tuple corresponds to one visited node along the walk. Hence, this operation will not incur any communication among the machines, and only has a constant expansion of the space cost, which has a negligible effect due to Lemma~\ref{lemma:pre}. 
We denote the set of all the decomposed tuples by $D_G$.

   





\setlength{\textfloatsep}{0pt}

\subsection{Detecting Meeting-Walks} \label{subsec:Detect Meeting Walks}
As the decomposed tuples are stored at different machines, challenges exist if we use these tuples to detect whether some pairs of random walks meet. Algorithm \ref{alg:meeting} shows the pseudo-code of detecting walk-meetings and computing SimRank values.

\vspace{1mm}
\noindent
{\bf Sorting in $O(1)$ Rounds.} Our first step is to sort all the tuples across the machines (Line 1). The sorting is based on the {\it multi-dimensional} sorting because the tuple contains five elements, i.e., sorting is first done based on the first element, and then elements that share the same first element, sorting is based on the second element, and so on. Elements are sorted based on their node IDs. We also intentionally let the node ID of source node $s$ be \textbf{smaller} than other nodes to guarantee that walks sharing the same first three elements start with the walk sourced at $s$. 

Recall in Section~\ref{sec:overall} we show that if two walks meet, then there must be two decomposed tuples, each from one walk, sharing the first three elements. After sorting the tuples, the tuples sharing the first three elements will be clustered together and the tuple whose fifth element is $s$, if exists, will be placed in the first position in the cluster. Particularly, suppose we have tuples $(j, i, v, l_1, s)$, $(j, i, v, l_2, u_1)$, $(j, i, v, l_3, u_2)$ that share the first three elements, implying that the $j$-th walk started from $s$ meets the $j$-th walks started from $u_1$ and $u_2$. These three tuples will be clustered in the order of $(j, i, v, l_1, s)$, $(j, i, v, l_2, u_1)$, $(j, i, v, l_3, u_2)$ after sorting. The formal result is as follows.

\begin{lemma}\label{lem:predecessor}

If the $j$-th walk $W$ started from the query source node $s$ meets the $j$-th walk $W'$ started from $u$, then there must be a decomposed tuple $T$ of $W$ and a decomposed tuple $T'$ of $W'$ that share the first three elements. Furthermore, after sorting the tuples, $T$ is the closest tuple that is decomposed from $W$ and comes before $T'$.
\end{lemma}

\noindent
{\bf Walk Paired-Up in $O(1)$ Rounds.}
To compute the SimRank between $s$ and $u_1$, we need to pair up $(j, i, v, l_1, s)$ and $(j, i, v, l_2, u_1)$; similarly, we also need to pair up $(j, i, v, l_1, s)$ and $(j, i, v, l_3, u_2)$ to compute the SimRank between $s$ and $u_2$. 
Pairing-up can be challenging because the tuples are stored at different machines, and hence communication between machines is unavoidable. In the MPC model, a machine that communicates with other machines would require the same space as the size of communication messages. To reasonably bound the communication cost between the machines, we employ the following \textit{PREDECESSOR} procedure \cite{DBLP:conf/focs/BehnezhadDELM19} to couple the tuples sharing the same first three elements, which can be finished in $O(1)$ communication rounds.


\begin{tcolorbox}
\textit{{\bf PREDECESSOR}: Considered an ordered list of tuples such that each tuple is labeled by 0 or by 1. Then, for each tuple $T^{\prime}$ labeled by 0, PREDECESSOR associates the closest tuple $T$ labeled by 1 such that $T$ comes before $T^{\prime}$ in the ordering. PREDECESSOR can be implemented in $O(1)$ MPC rounds with $n^{\alpha}$ space per machine, for any constant $\alpha>0$.}
\end{tcolorbox}

To apply \textit{PREDECESSOR} for pairing up the tuples, each machine can run a local algorithm to distinguish the tuples that are generated from the random walks starting at node $s$ and others. As shown in Figure \ref{fig:meeting}, all machines assign label $1$ to the tuples whose fifth elements are $s$ in parallel. Otherwise, the tuples are labeled by $0$. By Lemma~\ref{lem:predecessor}, we can apply the \textit{PREDECESSOR} procedure to associate the closest Label-$1$ tuple (i.e., the tuple generated from the walk sourced at $s$) to each of its following Label-$0$ tuples. Here, the {\it association} means that the machine holding the Label-$0$ tuple is aware of its closest Label-$1$ tuple before it. For each such associated pair $T$ and $T^{\prime}$, if their first three elements are the same, then the corresponding walks of $T^{\prime}$ and $T$ meet. 

\begin{algorithm}[t]
\small
\SetKwInOut{Input}{Input}\SetKwInOut{Output}{Output}
  \Input{Tuple set $D_G$ from Section~\ref{subsec:Decompose Random  Walks}; source node $s$}
  \Output{Estimated SimRank score $\tilde{s}(s, u)$ for each $u\in V$}
\emph{Sort all elements in $D_G$}\;
\emph{Calculate $z=\frac{\log{2n}}{2(\epsilon-\frac{3}{n})^2}$, $|{Z}_{u, l}| =  z\cdot (\sqrt{c})^{l}\cdot (1-\sqrt{c})$ and $|\hat{Z}_{u, l}| = \lceil |{Z}_{u, l}| \rceil$ for $l = 1,..., \log_{\frac{1}{\sqrt{c}}} n$} and $u\in V$ \;
{
\For{all tuple $(i, j, v_j, l, v)\in D_G$ in parallel}
{
\If{$v = s$}{
Label $(i, j, v_j, l, v)$ as $1$}
\Else{
Label $(i, j, v_j, l, v)$ as $0$}

}
}
{
\For{all tuple $(i, j, v_j, l_2, u)$ labeled by 0 in parallel}{
Link it to the closest tuple $(i, k, v_k, l_1, s)$ labeled by 1 using {PREDECESSOR}\;
\If{$(j=k) \wedge (v_j=v_k)$}{
$\tilde{s}(s,u)\gets \tilde{s}(s,u)+ \frac{|{Z}_{s, l_1}|\cdot |{Z}_{u, l_2}|}{|\hat{Z}_{s, l_1}|\cdot |\hat{Z}_{u, l_2}|} \cdot \frac{1}{z}$}
}
}
  \caption{Detect meeting and calculate SimRank}
  \label{alg:meeting}
\end{algorithm}

\subsection{Computing SimRanks}\label{subsec:Calculate SimRank}

The paired-up tuples in the previous step will contribute to the corresponding SimRank value. For example, suppose two walks from $s$ and $u$ meet, then this pair of walks contribute to the SimRank value $s(s,u)$, because $s(s,u)$ is estimated by the meeting probability of the walks from node $s$ and node $u$. The challenge is how to concurrently distribute and aggregate these {\it SimRank contributions} across multiple machines.

It is crucial to let those {\it SimRank contributions} corresponding to the same pair of nodes (e.g., $(s,u)$) be processed in the same machine, to avoid repeat counting of walk-meets. For example, walks $\{s, u_1, u_2, w\}$ and $\{v, u_1, u_2, k\}$ meet at both node $u_1$ and node $u_2$. To ensure accuracy, the event of each walk-meet should be counted only once. 
For this purpose, we let each machine specifically handle SimRank evaluations of $\frac{n}{M}$ nodes with respect to $s$. 
Without loss of generality, we assume the node IDs are from $1$ to $n$. Then, the $i$-th machine handles the node set with node IDs in $\left[\frac{n(i-1)}{M}+1, \frac{ni}{M}\right]$.

\vspace{1mm}
\noindent
{\bf Computing SimRanks in $O(1)$ Rounds.}
The detailed MPC operations of computing SimRanks are performed as follows (also illustrated in Algorithm~\ref{alg:meeting} Lines 8 - 11). Using the PREDECESSOR procedure, for each tuple $(i, j, v_j, l_2, u)$ with label $0$, it can be linked to the closest tuple $(i, k, v_k, l_1, s)$; if $j=k$ and $v_j=v_k$, then that indicates the meet of two walks represented by the two tuples. Then, a {\it walk-meet message} $(j, s, l_1, u, l_2)$ will be sent (by the machine holding tuple $(j, i, v, l_1, s)$) to the target machine that is responsible to computing SimRank $s(s,u)$. Here, the walk-meet message $(j, s, l_1, u, l_2)$ indicates that a length-$l_1$ walk sourced at node $s$ and a length-$l_2$ walk sourced at node $u$ meet. Also, assuming $s$, $u$, and $v$ are node IDs from $1$ to $n$, the target machine that the message is sent to is the $\left\lfloor \frac{uM}{n} \right\rfloor$-th machine because this machine will compute $s(s,u)$. Each machine may then receive multiple copies of the walk-meet message $(j,s,l_1, u,l_2)$, and only one of them will be kept because each walk-meet event should be counted only once. For each unique message $(j,s,l_1, u, l_2)$ received, we add a value $\frac{|{Z}_{s, l_1|}\cdot |{Z}_{u, l_2}|}{|\hat{Z}_{s, l_1}|\cdot |\hat{Z}_{u, l_2}|} \cdot \frac{1}{z}$ to the SimRank value $\tilde{s}(s,u)$, where $z=\frac{\log{2n}}{2(\epsilon-\frac{3}{n})^2}$ is the expected number of walk samples, $\hat{Z}_{s, l}$ denotes the set of actually generated length-$l$ random walks
sourced at $s$, and ${Z}_{s, l}$ is the expected set of length-$l$ random walks to be generated sourced at $s$. 

We distinguish actual and expected samples because ${Z}_{s, l}$ may not be an integer. Particularly, we note that $|\hat{Z}_{s, l}| =  \left\lceil z\cdot (\sqrt{c})^{l}\cdot (1-\sqrt{c})\right\rceil$, and $|{Z}_{s, l}| =   z\cdot (\sqrt{c})^{l}\cdot (1-\sqrt{c})$. When $z\cdot (\sqrt{c})^{l}\cdot (1-\sqrt{c})$ is not an integer, the expected walk set ${Z}_{s, l}$ is defined based on size-$\lceil|{Z}_{s, l}|\rceil$ set $Z_1$ and size-$\lfloor|{Z}_{s, l}|\rfloor$ set $Z_2$, such that there is probability $|{Z}_{s, l}|-\lfloor|{Z}_{s, l}|\rfloor$ that one walk of ${Z}_{s, l}$ is selected from set $Z_1$ and otherwise select from set $Z_2$. 


\vspace{1mm}
\noindent
{\bf Space Analysis for Walk-Meet Messages.} 
The walk-meet messages received by each machine are at most $\log_{{1}/{\sqrt{c}}}n$ times the number of the random walks started at the nodes handled by the machine. To see this, in the worst case, for each node of a walk $W$ generated from $u$ handled by the machine, there is a walk $W'$ from $s$ that meets $W$ at the node, creating a walk-meet message. Since each machine is responsible to compute SimRank values for $n/M$ nodes, the space per machine for receiving the message is $$O\left(\frac{n}{M}\log_{\frac{1}{\sqrt{c}}}n\cdot \frac{\log 2n}{2(\epsilon-\frac{3}{n})^2}\right)=O\left(\frac{n\log^2 n}{M(\epsilon-\frac{3}{n})^2}\right)$$ 
The total space cost for walk-meet messages is therefore $$O\left(\frac{n\log^2 n}{M(\epsilon-\frac{3}{n})^2}\cdot M\right)=O\left(\frac{n\log^2 n}{(\epsilon-\frac{3}{n})^2}\right)$$

\section{Summary of the Results}\label{sec:correctness}
\noindent
{\bf Round Complexity}. The round complexity is dominated by the generation of parallel random walks, which is $O(\log^2\log n)$ by Theorem~\ref{thm:walk}.  

\vspace{1mm}
{\noindent}
{\bf Space Complexity}. The total space cost for walk-meet messages is dominated by the random-walk generation space cost $O\left(m \log n+n^{1+o(1)} \log ^{4.5} n + \frac{n \log ^{3+o(1)} n}{2\left(\epsilon-\frac{3}{n}\right)^{2}}\right)$ (see Theorem~\ref{thm:walk}). 
We note that the total space cost is $\tilde{O}(m+n)$, indicating the existence of $M=\tilde{\Theta}(\frac{m+n}{S})$ for $S=n^{\alpha}$.

All these results give us Theorem~\ref{thm:main}, which also supports Theorem~\ref{the:contri}.

\begin{theorem}\label{thm:main}
Given a graph of $n$ nodes and $m$ edges, and a constant error $\epsilon>0$, for sufficiently large $n$, Algorithm~\ref{alg:overall} can be run in MPC using $O\left(\log^2 \log n\right)$ communication rounds, with a total space per-round as $O\left(m\log n+n^{1+o(1)}\log^{4.5} n+\frac{n\log^{3+o(1)}n}{2(\epsilon-\frac{3}{n})^2}\right)$. The space per-machine needed is $S=n^{\alpha}$ such that $n^{\alpha}>{\log^5 n}/({\log\frac{1}{\sqrt{c}}})$. Algorithm~\ref{alg:overall} is an imperfect sampler that fails with probability at most $O\left(n^{-1}\right)$.
\end{theorem}



\vspace{1mm}
\noindent
{\bf Accuracy Analysis}. The accuracy guarantee of our random-walk sampling techniques resembles that of a simple Monte Carlo method, but we also incorporate the techniques of length-truncation and the rounding of the number of walks. These brings sophistication in the proof which we leave to the Appendix, and only include the main result as follows. We note that the threshold $\frac{3}{n}$ in the result can be further reduced (see Appendix).
\begin{theorem}\label{thm:error}
Algorithm~\ref{alg:overall} outputs SimRank values $\tilde{s}(s,u)$ with error at most $\epsilon$ ($\epsilon \ge \frac{3}{n}$) and with probability at least $1-\frac{1}{n}$.
\end{theorem}

\bibliography{MPCSimRank}
\bibliographystyle{unsrt}

\newpage
\appendix

\section{Missing Complexity Analysis}\label{appen:baseline note}
In Table \ref{tab:complexity}, we explore the literature and provide analytical results regarding space cost and round complexity. We note that all these three algorithms aim to sample random walks in parallel for calculating the endpoint distribution (CloudWalker) or the meeting probability of two paths (UniWalk and DISK). Following existing analysis in UniWalk~\cite{DBLP:journals/tkde/SongLGZWY18}, we assume one random walk step in a distributed system costs one communication round.

\subsection{CloudWalker~\cite{DBLP:journals/pvldb/LiFLCCL15}} CloudWalker uses the following definition to compute SimRank:
\begin{align}\label{equ:truncatedD}
   {\bm S}=c{\bm P}^{\top}{\bm S}{\bm P}+{\bm D}=\sum_{t=0}^{\infty} c^t {\bm P}^{\top t} {\bm D} {\bm P}^t, 
\end{align}
where ${\bm S}$ is the SimRank matrix that ${\bm S}(i,j)$ is SimRank score between the $i$-th node and $j$-th node, ${\bm P}$ is the transition matrix of $G^T$ which is a transpose of the input graph $G$, and ${\bm D}$ is the diagonal correction matrix that CloudWalker aims to compute first. The method of CloudWalker mainly includes the \textit{calculation of ${\bm D}$} and \textit{calculation of single-source SimRank}.

\noindent
\textbf{Calculation of ${\bm D}$.} CloudWalker relies on a method that combines the Jacobi method and Monte Carlo method to calculate the diagonal correction matrix ${\bm D}$ whose each row corresponds to the correction values of one node to the other nodes. 
Under the assumption that the input graph cannot be fit in a single machine ($S=n^{\alpha}$, $0<\alpha<1$),  CloudWalker cannot process all nodes at a time and adopts the distributed structure to store the data for each node. Particularly, each machine will process $b$ nodes at a time and $b$ is assigned a node number to be processed. As a result, CloudWalker needs to perform its algorithms for $\frac{n}{b}$ times to process all nodes.
We note that in the process of estimating {${\bm D}$}, CloudWalker needs to sample $N = {(\log n)}/{\epsilon_p^2}$ length-$l$ random walks for each node, where $\epsilon_p$ is a parameter that controls the accuracy. When sampling $N$ random walks, as we assume machine space $S<n$, each machine may need communication with other machines to access the adjacent nodes which might be stored in other machines. For example, assuming that $u\to v$ is an edge in $E$ and $u$, $v$ are stored in two different machines, a walk-step $u\to v$ requires one communication round. 
Therefore sampling a length-$l$ random walk needs $O(l)$ communication rounds. After running $I$ iterations for ${n}/{b}$ times, the overall communication rounds of CloudWalker is $O({nIl}/{b})$.

Then in the Jacobi method, CloudWalker needs $O(blN)$ space to store the walk state matrix \textbf{$A$} in every single machine. In addition, they need to store all the nodes and edges, incurring $O(m+n)$ costs. Putting together, the space cost per round is $O\left({(bl\log n)}/{\epsilon_p^{2}}+n+m\right)$. Furthermore, each machine takes $O\left({(bl\log n)}/{\epsilon_p^{2}}\right)$ space to store $A$.
Assuming the graph is equally distributed across the machines, then the per-machine space is $O\left({(bl\log n)}/{\epsilon_p^{2}}+(m+n)/M\right)$.

\noindent
\textbf{Calculation of single-source SimRank.}
In Algorithm 4 of \cite{DBLP:journals/pvldb/LiFLCCL15}, CloudWalker conducts the computation of Single-source SimRank based on the obtained diagonal correction matrix ${\bm D}$. In Lines 4-13, CloudWalker again utilizes the random walks to compute the vectors corresponding to SimRank results, which needs to consume $O(l^2)$ communication rounds. The rationale of the communication round consumption is also that walk sampling requires to access other nodes which may be resided on multiple machines. Similar to the computation of ${\bm D}$, $b$ nodes will be processed at a time and the number of overall communication rounds is $O(\frac{nl^2}{b})$. For the reason that $O(bl^2N)$ space needs to be used for storing the computation vectors related to random walks, the total space complexity is $O\left({(bl^2\log n)}/{\epsilon_p^{2}}+n+m\right)$.

To conclude, CloudWalker needs $O({nl(I+l)}/{b})$ communication rounds and $O\left({(bl^2\log n)}/{\epsilon_p^{2}}+n+m\right)$ total space complexity, where each machine requires $O\left({(bl^2\log n)}/{\epsilon_p^{2}}+(m+n)/M\right)$ space. Furthermore, CloudWalker needs to make a trade-off between the per-machine-space memory cost and the number of communication rounds. A special case is when $b$ equals $n$, CloudWalker can finish the SimRank within $O({Il+l^2})$ rounds where each machine requires  $O(n \log n)$ space. Finally, CloudWalker only provides the error bound when estimating {${\bm D}$} and truncating the length $l$; they did not show any accuracy guarantees for computing SimRank scores \cite{DBLP:journals/pvldb/0012XFC00M20}. 

\subsection{UniWalk~\cite{DBLP:journals/tkde/SongLGZWY18}} The application scope of UniWalk is {\it undirected graph}. To compute the SimRank score $s(u,v)$, Uniwalk adopts the Monte Carlo method by sampling a certain number of Length-$L$ random walks and calculating the expected meeting distance of two bidirectional Length-$L$ random walks from two nodes $u$ and $v$.
The main contribution of UniWalk is to convert sampling $O(nN)$ bidirectional Length-$L$ walks to only $O(N)$ Length-$2L$ unidirectional walks. 
A rectified factor is introduced to make the two kinds of walks equivalent. $L$ is empirically set in UniWalk, and UniWalk takes $l=2L$ rounds to generate length-$2L$ unidirectional walks. By sampling $N = \frac{c^2(1-c^L)^2 d_{\text{max}}^2 \log 2n}{2\epsilon^2 (1-c)^2 d_{\text{min}}^2}$ random walks, where $d_{\text{max}}$ and $d_{\text{min}}$ are the maximum and minimum degree in $G$ respectively. UniWalk can guarantee with high probability that the computed SimRank values 
have at most $\epsilon$ error. 
In the distributed version of UniWalk (Algorithm 3 in~\cite{DBLP:journals/tkde/SongLGZWY18}), to guarantee $O(l)$ communication rounds, we need to apply Algorithm 3 $N$ times in parallel. Then in the worst case, Line 21 generates $N$ iterations and Line 23 needs to send a message of size $O(l)$. We note that $O(Nl)$ is already the worst-case total message space cost generated during the UniWalk process. 
Hence, both the per-machine space cost and total space cost are $O(Nl) = O(\frac{n^2 l\log n}{\epsilon^2})$ because $d_{max} = O(n)$ and $d_{min} \ge 1$.

\subsection{DISK~\cite{DBLP:journals/pvldb/0012XFC00M20}} Similar to CloudWalker, DISK also utilizes random walks to estimate the diagonal correction matrix {${\bm D}$} and then performs the single-source SimRank computation. We also demonstrate the processes from the perspective of the \textbf{calculation of} {${\bm D}$} and \textbf{calculation of single-source SimRank}.

\noindent
\textbf{Calculation of {${\bm D}$}.}
The following analysis is based on the results in Section~4 of~\cite{DBLP:journals/pvldb/0012XFC00M20}. DISK explores the physical meaning of {${\bm D}$} and conducts a more efficient algorithm to estimate the diagonal correction matrix {${\bm D}$}. Particularly, it transfers the calculation of {${\bm D}$} into the meeting probability of two $\sqrt{c}$-walks starting from each node $v\in V$. The main advantage of this method is that the estimation accuracy of {${\bm D}$} can be guaranteed. Then DISK build a tree-based method to sample $N =O(\frac{\log n}{\epsilon_d^{2}})$ trees when we set the failure probability $\delta_d = \frac{1}{n}$ for each node $v\in V$. The trees can also be interpreted as the $\sqrt{c}$-walks.
Since the height of the trees is expected to be $O(\log n)$ and each level of one tree will require one communication round, the overall communication rounds after building $N$ trees can be $O(\frac{\log^2 n}{\epsilon_d^{2}})$. For the space cost, DISK needs to save the entire tree (random walk), which is $O(\log n)$, and the space per machine can be $O(\frac{m\log n}{M})$.

\noindent
\textbf{Calculation of single-source SimRank.}
The final single-source computation is based on the truncated formulation of Equation \ref{equ:truncatedD} as:
\begin{align}\label{equ:trauncatedD2}
    \textbf{S}\approx\sum_{t=0}^{K} c^t {\bm P}^{\top t} {\bm D} {\bm P}^t,
\end{align}
where $K$ notes the truncated terms for approximation.
In order to obtain the final SimRank score, $O(K)$ communication rounds will be consumed to calculate the result of Equation \ref{equ:trauncatedD2}. For the reason that each node also needs to maintain $O(K)$ space for the above equation, the space for each machine will be $O(\frac{Kn}{M})$.

Hence we summarize that the number of overall communication rounds for DISK is $O\left(\frac{\log^2 n}{ \epsilon_d^{2}}+K\right)$. Then in the space analysis, DISK needs total space $O(m\log n +Kn)$ and space $O(\frac{m\log n+Kn}{M} )$ per machine to store the trees. The estimation error of DISK is controlled as smaller than $\frac{c\left(1-c^{K}\right) \epsilon_d}{1-c}+c^{K+1}$ with high probability, which corresponds to our statistic in Table \ref{tab:complexity}.

\section{Missing Proofs}
\subsection{Proof of Theorem \ref{thm:walk}}\label{appen:walk proof}
\begin{proof}
For any instance $n$, let $\beta$ satisfy that $n^{\beta}=n^{\alpha}/ \log_{\frac{1}{\sqrt{c}}}n$. If $n^{\beta}>{(\log^3n \cdot \log_{\frac{1}{\sqrt{c}}}n)}$ (equivalently, $n^{\alpha}>{\log^4 n\log _{\frac{1}{\sqrt{c}}}n}$=$\frac{\log^5 n}{\log\frac{1}{\sqrt{c}}}$), then $\log_{\frac{1}{\sqrt{c}}}n=o(S)/\log^3 n$, and hence Theorem~\ref{stoc} is applicable for $S=n^{\beta}$. Applying Theorem~\ref{stoc} with $S=n^{\beta}$ still gives us $O(\log^2\log n+\log^2l)$ round complexity and $O(m+n^{1+o(1)}l^{3.5}+Nnl^{2+o(1)})$ space because the two complexities hides the $\alpha$ (or $\beta$) factor. Also, when applying Theorem~\ref{stoc} with $S=n^{\beta}$, the number of machines $M$ satisfies that $M\cdot n^{\beta}=\tilde{\Theta}(m+n)$ by the requirement of the MPC model. Then, $n^{\alpha}\cdot M=\tilde{\Theta}(m+n)$ also holds, which means the number of machines used is desired for MPC when $=n^{\alpha}$. Now, for each of the $M$ machines each having space $n^{\alpha}$, we divide the machine space into $\log_{\frac{1}{\sqrt{c}}n}$ subspaces, and for each subspace we apply WRME algorithms (Theorem~\ref{stoc}) in parallel, respectively for $l=1, 2, \ldots, \log_{\frac{1}{\sqrt{c}}n}$. 

This indicates that each WRME algorithm runs with $S=n^{\beta}$. The worst-case round complexity corresponds to the WREM algorithm running with the longest length, which is $\log_{\frac{1}{\sqrt{c}}} n$. By Theorem~\ref{stoc}, we have the round complexity 
\begin{align}
O\left(\left\lceil\log ^{2} \log n+\log ^{2} l \right\rceil\right)   = & O\left(\left\lceil\log ^{2} \log n+\log ^{2} \log_{\frac{1}{\sqrt{c}}} n \right\rceil\right) \nonumber \\=  &O\left(\log^2\log n\right)
\end{align}
The last step holds because $c$ is a constant.
For each WRME instance, a machine space $n^{\beta}$ is sufficient. Considering all $\log_{\frac{1}{\sqrt{c}}n}$ WRME instances run in parallel, the space per machine is $n^{\beta}\cdot \log_{\frac{1}{\sqrt{c}}n}=n^{\alpha}$.

For the total space cost, we need to sum up all the space costs for each call of the WRME algorithm. Therefore, the space bound is calculated by applying Theorem~\ref{stoc} for $1\leq l \leq \log_{\frac{1}{\sqrt{c}}} n$, and for each $l$, the space cost is expanded $\frac{\log{2n}}{2\left(\epsilon-\frac{3}{n}\right)^2}\cdot (\sqrt{c})^{l}\cdot (1-\sqrt{c})$ times (the number of random walks):
\hspace{-8mm}
{\small
\begin{align}\label{equ:walk cost}
    &\sum_{l=1}^{\log_{\frac{1}{\sqrt{c}}} n } O\left(m+n^{1+o(1)} l^{3.5}+N_l n l^{2+o(1)}\right) \nonumber \\
    = & \sum_{l=1}^{\log_{\frac{1}{\sqrt{c}}} n} 
    O\left(m+n^{1+o(1)} l^{3.5}+\right. \left.\frac{\log{2n}}{2(\epsilon-\frac{3}{n})^2}\cdot (\sqrt{c})^{l}\cdot (1-\sqrt{c}) n l^{2+o(1)}\right)
    \nonumber\\
    = & O\left(m\log n+n^{1+o(1)}\log^{4.5} n+\frac{n\log^{3+o(1)}n}{2(\epsilon-\frac{3}{n})^2}\right)
\end{align}
}
The last step holds because $\sum_{l=1}^{\log_{\frac{1}{\sqrt{c}}}n}(\sqrt{c})^{l}(1-\sqrt{c})< 1$. 
The space cost per machine is expanded at most $\log_{\frac{1}{\sqrt{c}}}n$ compared with the per-machine cost needed for Theorem~\ref{stoc}, and thus still being strongly sub-linear to $n$.
\end{proof}

\subsection{Proof of Lemma~\ref{lemma:pre}}

\begin{proof}
Consider that there exists a distributed algorithm $\mathcal{A}$ that works on $M=\tilde{O}(\frac{m+n}{n^{\alpha}})$ machines with per-machine space $S={\Theta}(n^{\alpha})$ for $0<\alpha < 1$ using $O(g(\alpha)\cdot f(n))$ communication rounds, where $g(\alpha)$ is a constant function that is computed based on $\alpha$. 
By the definition of big-$O$ notation, there must exist constants $C$, $C'$ and $n_0$, such that for any $n'>n_0$, the per-machine space $S$ used in $\mathcal{A}$ is at most $C\cdot n'^{\alpha}$, and the number of rounds used in $\mathcal{A}$ is at most $C'\cdot g(\alpha)\cdot f(n')$. 

Let $\beta=\log_{n'}{(C\cdot n'^{\alpha})}$, and we have $n'^{\beta}=C\cdot n'^{\alpha}$. Hence, Algorithm $\mathcal{A}$ is conducted on machines with space $n'^{\beta}$ using at most $C'\cdot g(\alpha)\cdot f(n')$ 
rounds. The number of machines $M=\tilde{O}(\frac{m+n}{n'^{\alpha}})=\tilde{O}(\frac{m+n}{C\cdot n'^{\alpha}})=\tilde{O}(\frac{m+n}{n'^{\beta}})$. Hence, Algorithm $\mathcal{A}$ is a valid MPC algorithm for $S=n'^{\beta}$ when the number of graph nodes is $n'$, and the number of rounds used is a constant factor multiplied by $f(n')$. 

We fix $\beta$ and let $M^*$ be the maximum number of machines used considering all $n'>n_0$, and it is easy to see $M=\tilde{O}(\frac{m+n}{n^{\beta}})$.

Meanwhile, note that $\beta=\alpha+\log_{n'}(C')$. When $n'$ is sufficiently large, $\beta$ can be arbitrarily close to $\alpha$. Therefore, for any $\beta \in (0,1)$ we can always find $\alpha = \beta-\log_{n'}(C')>0$ for sufficiently large $n'$, and hence, we can construct the corresponding MPC algorithm $\mathcal{A}$. 

We can then safely replace symbol $\beta$ with $\alpha$ and the lemma holds. 
\end{proof}

\subsection{Proof of Lemma \ref{lem:predecessor}}\label{appen:lemma2 proof}

\begin{proof}
Suppose Walk $W$ and Walk $W'$ meet at their $i$-th step at node $v$, then there is a decomposed tuple $T=(j, i, v, l_1, s)$ from $W$, and a decomposed tuple $T'=(j, i, v, l_2, u)$ from $W'$. They share the first three elements. Furthermore, we prove by contraction that there is a decomposed tuple $T^*$ originated from $W$ but $T^*$ is closer to $T'$ in the sorting order. Since the first three elements of $T$ and $T'$ are the same, then $T^*$ must also have the same first three elements. However, both $T$ and $T^*$ are decomposed from the walk $W$, but any two tuples decomposed from the same walk must have different elements in the first three positions of the tuple. Contradiction ensues.
\end{proof}

\subsection{Proof of Theorem~\ref{thm:main}}
\begin{proof}
The proof is immediate from Theorem~\ref{thm:walk} and all the analysis results from Section~\ref{sec:MPC_analysis}.
\end{proof}

\subsection{Proof of Theorem \ref{thm:error}}\label{appen:theo4 proof}
We aim to show that our algorithm can obtain the SimRank values with at most $\epsilon$ absolute error compared with the true value. Recall that we sample $N_l = \left\lceil \frac{\log{2n}}{2\left(\epsilon-\frac{3}{n}\right)^2}\cdot (\sqrt{c})^{l}\cdot (1-\sqrt{c})\right\rceil$ length-$l$ ($ 1\leq l \leq\log_{\frac{1}{\sqrt{c}}}$) walks for each node. We will utilize two subsections to demonstrate how we bound the estimation error.

\subsubsection{Controlling Truncated Error}
Based on the Monte Carlo method, we will show that checking the meeting probability by sampling truncated $\sqrt{c}$-decay walks from $s$ and $u$ gives an $\epsilon$ ($\epsilon \ge \frac{3}{n}$) absolute error. The derivation is based on the Hoeffding Inequality as follows:
\begin{lemma}
Let $X_1, \ldots, X_z$ be independent bounded random variables with $X_i \in [0, 1]$ for all $i$. Then

\begin{align}
    {\bm\Pr}\left[\left|\frac{1}{z} \sum_{i=1}^z\left(X_i-E\left[X_i\right]\right)\right| \geq \lambda\right] \leq 2e^{\left(-2 z \lambda^2\right)}.
\end{align}

\end{lemma}

We note that the probability that a pair of $\sqrt{c}$ walks from node $s$ and node $u$ meet is the SimRank value $s(s,u)$, and we let $X_i$ be the event that the $i$-th pair $\sqrt{c}$-walks from $s$ and $u$ meet. Hence, $\mathbb{E}[X_i]=s(s,u)$. By simple application of Hoeffding Inequality and the estimation $\tilde{s}(s,u)=\frac{\sum_{i=1}^z X_i}{z}$, we have

\begin{align}\label{equ:hoeff}
&{\bm\Pr}\left[\left|\tilde{s}(s,u)-{s}(s,u)\right|\ge \lambda\right]\nonumber\\
=&{\bm \Pr}\left[\left|\frac{\sum_{i=1}^z X_i}{z}-\frac{\sum_{i=1}^z E(X_i)}{z}\right|\ge \lambda\right]
\leq  2e^{-2z\lambda^2}
\end{align}

Let $z=\frac{\log 2n}{2\lambda^2}$, we have $e^{-2z\lambda^2}=\frac{1}{n}$. However, our algorithm does not involve the $\sqrt{c}$-walks that are longer than $\log_{\frac{1}{\sqrt{c}}} n$. To quantify the effect of truncation, we note that sampling a $\sqrt{c}$-decay walk is equivalent to sampling a length-$l$-walk with probability $(\sqrt{c})^{l}(1-\sqrt{c})$ for $l\ge 0$. By calculating the probability where a $\sqrt{c}$-walk is less than $\log_{\frac{1}{\sqrt{c}}} n$, we have: 
\begin{align}\label{equ:small percent}
{\bm\Pr}\left[l \leq \log_{\frac{1}{\sqrt{c}}} n\right]
&=\sum_{l=0}^{\log_{\frac{1}{\sqrt{c}}} n} (\sqrt{c})^{l}\cdot (1-\sqrt{c})>1 - \frac{1}{n}
\end{align}
The above equation demonstrates that the probability of a random walk longer than $\log_{\frac{1}{\sqrt{c}}} n$ is smaller than $\frac{1}{n}$. Then we set $N_l = \left\lceil \frac{\log{2n}}{2\lambda^2}\cdot (\sqrt{c})^{l} \cdot (1-\sqrt{c})\right\rceil$ for length-$l$ ($ 0\leq l \leq\log_{\frac{1}{\sqrt{c}}}n$) based on this probability result, and truncate the length at $\log_{\frac{1}{\sqrt{c}}} n$. Next we show that this simplification will impact less than $\frac{3}{n}$ absolute value of the SimRank value. 

We denote ${C}_{s,l_1,u,l_2}$ as the number of meets between $s$'s length-$l_1$ walks and $u$'s length-$l_2$ walks,
if we directly perform the random walks based on the expected numbers, i.e., $|Z_{s,l}|$. Let $\phi = \log_{{1}/{\sqrt{c}}} n$.
We can transform Equation \ref{equ:hoeff} as follows:
{\small
\begin{align}\label{equ:7}
&{\bm \Pr}\left[|\tilde{s}(s,u)-{s}(s,u)|< \lambda\right]\nonumber\\
=&{\bm \Pr}\left[ -\lambda<\frac{\sum_{l_1=0}^{+\infty} \sum_{l_2=0}^{+\infty} {C}_{s,l_1,u,l_2} }{z}-{s}(s,u)
< \lambda \right]\nonumber\\
=&{\bm \Pr}\left[-z\lambda<{\sum_{l_1=0}^{\phi} \sum_{l_2=0}^{\phi} {C}_{s,l_1,u,l_2} }\right.\nonumber
+\left.{\sum_{l_1=0}^{\phi} \sum_{l_2=\phi+1 }^{+\infty} {C}_{s,l_1,u,l_2} }+\right.\\
&\left.{\sum_{l_1=\phi+1}^{+\infty} \sum_{l_2=0}^{\phi} {C}_{s,l_1,u,l_2} }\right.\nonumber
+\left.{\sum_{l_1=\phi+1}^{+\infty} \sum_{l_2=\phi+1}^{+\infty} {C}_{s,l_1,u,l_2} }-s(s,u)\cdot z< z\lambda\right]\nonumber\\
\ge & 1-2e^{-2z\lambda^2}
\end{align}
}
By Equation \ref{equ:small percent}, there are fewer than $\frac{z}{n}$ random walks sourced at nodes $s$ (resp. $u$) whose lengths are longer than $\phi = \log_{{1}/{\sqrt{c}}} n$. Hence, there are at most $\frac{z}{n}$ meets of paired-walks when one walk is longer than $\phi = \log_{{1}/{\sqrt{c}}} n$. Formally, $$\sum_{l_1=0}^{\phi} \sum_{l_2=\phi+1 }^{+\infty} {C}_{s,l_1,u,l_2}\leq\frac{z}{n}$$
$$\sum_{l_1=\phi+1}^{+\infty} \sum_{l_2=0}^{\phi} {C}_{s,l_1,u,l_2}\leq\frac{z}{n}$$
$$\sum_{l_1=\phi+1}^{+\infty} \sum_{l_2=\phi+1}^{+\infty} {C}_{s,l_1,u,l_2} \leq\frac{z}{n}$$
Then we have


{\small
\begin{align}\label{equ:8}
    &0 < \frac{\sum_{l_1=1}^{\phi} \sum_{l_2=\phi+1 }^{+\infty} {C}_{s,l_1,u,l_2} }{z}+\frac{\sum_{l_1=\phi+1}^{+\infty} \sum_{l_2=1}^{\phi} {C}_{s,l_1,u,l_2} }{z}\nonumber \\
    &+ \frac{\sum_{l_1=\phi+1}^{+\infty} \sum_{l_2=\phi+1}^{+\infty} {C}_{s,l_1,u,l_2} }{z} < \frac{3}{n}.
\end{align}}

If we apply Equation \ref{equ:8} in Equation \ref{equ:7}, we have:

{\small
\begin{align}
    &{\bm \Pr}\left[|\tilde{s}(s,u)-{s}(s,u)|< \lambda\right]\nonumber\\
    \le & {\bm \Pr}\left[-\lambda-\frac{3}{n}<\frac{\sum_{l_1=0}^{\phi} \sum_{l_2=0}^{\phi} {C}_{s,l_1,u,l_2} }{z}-{s}(s,u)<\lambda\right]\nonumber\\
    \le& {\bm \Pr}\left[-\lambda-\frac{3}{n}<\frac{\sum_{l_1=0}^{\phi } \sum_{l_2=0}^{\phi} {C}_{s,l_1,u,l_2} }{z}-{s}(s,u)<\lambda+\frac{3}{n}\right]
\end{align}}
For the reason that we only measure the meetings happening at walk-lengths less than $\phi = {\log_{{1}/{\sqrt{c}}} n }$, we set $\tilde{s}(s,u) = \frac{\sum_{l_1=0}^{\phi} \sum_{l_2=0}^{\phi} {C}_{s,l_1,u,l_2} }{z}$ as our estimation result. Lastly we have the following accuracy guarantee:

\begin{align}
    &{\bm \Pr}\left[\left|\tilde{s}(s,u) - s(s,u)\right|<\lambda+\frac{3}{n}\right]\ge 1-2e^{-2z\lambda^2}
\end{align}
We simplify this result by setting $\lambda = \epsilon - \frac{3}{n}$, and we have:
\begin{align}\label{equ:failure_pro}
    &{\bm \Pr}\left[|\tilde{s}(s,u)-s(s,u)| \ge \epsilon\right] \le 2e^{-2z{(\epsilon-\frac{3}{n})}^2},
\end{align}
where $z = \frac{\log 2n}{2{\left(\epsilon-\frac{3}{n}\right)}^2}$ and $\epsilon \ge \frac{3}{n}$. Based on this result, we guarantee that by sampling $N_l = \left\lceil \frac{\log{2n}}{2\left(\epsilon-\frac{3}{n}\right)^2}\cdot (\sqrt{c})^{l}\cdot (1-\sqrt{c})\right\rceil$ length-$l$ ($ 0\leq l \leq\log_{\frac{1}{\sqrt{c}}}$) walks for each node, we can obtain the SimRank values $\tilde{s}(s,u)$ with at most $\epsilon$ error ($\epsilon \ge \frac{3}{n}$).

\begin{figure*}

\end{figure*}



\subsubsection{Rounding Up Error}
we consider that the actual number of random walks for each length $l$ is $\left\lceil \frac{\log{2n}}{2{\left(\epsilon-\frac{3}{n}\right)}^2}\cdot (\sqrt{c})^{l}\cdot (1-\sqrt{c})\right\rceil$ because $\frac{\log{2n}}{2{\left(\epsilon-\frac{3}{n}\right)}^2}\cdot (\sqrt{c})^{l}\cdot (1-\sqrt{c})$ may not be an integer. The rounding up of the samples for length-$l$ makes our SimRank estimation biased, which is denoted by a followed example.
We let $\hat{Z}_{s, l_1}$ denote the set of actually generated length-$l_1$ random walks
sourced at $s$, where $|\hat{Z}_{s, l_1}| =  \left\lceil \frac{\log{2n}}{2{\left(\epsilon-\frac{3}{n}\right)}^2}\cdot (\sqrt{c})^{l_1}\cdot (1-\sqrt{c})\right\rceil$. And we denote the expected set of length-$l_1$ random walks as ${Z}_{s, l_1}$ with $|{Z}_{s, l_1}| =   \frac{\log{2n}}{2{\left(\epsilon-\frac{3}{n}\right)}^2}\cdot (\sqrt{c})^{l_1}\cdot (1-\sqrt{c})$.
 When two random walks selected from $Z_{s, l_1}$ and $Z_{u, l_2}$ meet, rounding up $|{Z}_{s, l_1}|$ to $ |\hat{Z}_{s, l_1}|$  and $|{Z}_{u, l_2}|$ to $ |\hat{Z}_{u, l_2}|$ will change the distribution of walk length. As a result,
the meeting probability of $s$ and $u$ will be affected and degrade our accuracy. 
We give a result as follows:

\begin{lemma}\label{lem:meeting score}
Each time when two random walks selected from $\hat{Z}_{s, l_1}$ and $\hat{Z}_{u, l_2}$ meet, this meeting will only contribute $\frac{|{Z}_{s, l_1}|\cdot |{Z}_{u, l_2}|}{|\hat{Z}_{s, l_1}|\cdot |\hat{Z}_{u, l_2}|} \cdot \frac{1}{z}$ score to the SimRank value of $\tilde{s}(s,u)$.
\end{lemma}

\begin{proof}
When we estimate the SimRank $\tilde{s}(s,u)$ values based on the changed sample numbers $|\hat{Z}_{s, l_1}|$ and $|\hat{Z}_{u, l_2}|$, we still need to guarantee the meeting probability is the same as using the expected sample numbers $|{Z}_{s, l_1}|$ and $|{Z}_{u, l_2}|$:
\begin{equation}
    \frac{\hat{C}_{s,l_1,u,l_2}}{|\hat{Z}_{s, l_1}|\cdot |\hat{Z}_{u, l_2}|} =  \frac{{C}_{s,l_1,u,l_2}}{|{Z}_{s, l_1}|\cdot |{Z}_{u, l_2}|}.
\end{equation}
Here $\hat{C}_{s,l_1,u,l_2}$ means the actual meeting numbers where $s$' length-$l_1$ walks meet with $u$' length-$l_2$ walks when we run the actual random walks. The difference between $\hat{C}_{s,l_1,u,l_2}$ and ${C}_{s,l_1,u,l_2}$ is that ${C}_{s,l_1,u,l_2}$ means the meeting number when running the theoretical random walks.
Then we derive the SimRank value $\tilde{s}(s,u)$ according to:
\begin{align}\label{equ:scale contribution}
   \tilde{s}(s,u) =& \frac{\sum_{l_1=1}^{\log_{\frac{1}{\sqrt{c}}} n } \sum_{l_2=1}^{\log_{\frac{1}{\sqrt{c}}} n} {C}_{s,l_1,u,l_2} }{z}\nonumber \\
   = & \frac{\sum_{l_1=1}^{\log_{\frac{1}{\sqrt{c}}} n } \sum_{l_2=1}^{\log_{\frac{1}{\sqrt{c}}} n} \hat{C}_{s,l_1,u,l_2} \cdot \frac{|{Z}_{s, l_1}|\cdot |{Z}_{u, l_2}|}{|\hat{Z}_{s, l_1}|\cdot |\hat{Z}_{u, l_2}|}}{z}
\end{align}
Here $\hat{C}_{s,l_1,u,l_2}$ is derived by counting the times where walks from the actual sample sets $\hat{Z}_{s, l_1}$ and $\hat{Z}_{u, l_2}$ meet. This result proves the claim.

\end{proof}
\vspace{-4mm}
 \subsection{Reducing Threshold $\frac{3}{n}$ and Failure Probability $\frac{1}{n}$}
 
 In Theorem \ref{thm:error}, we note that we can further reduce the error threshold (e.g., $\frac{3}{n}$) and the failure probability (e.g. $\frac{1}{n}$) by increasing the length of
walks and enlarging the number of samples, respectively. 

\noindent
{\bf Reduce the error threshold.}
 Assuming we extend the maximum length to $p\cdot\log_{\frac{1}{\sqrt{c}}} n$, where $p$ is a constant and $p\ge 1$, then we can transform Equation \ref{equ:small percent} as following:
 \begin{align}
{\bm\Pr}\left[l \leq p\cdot\log_{\frac{1}{\sqrt{c}}} n\right]
&=\sum_{l=0}^{p\cdot\log_{\frac{1}{\sqrt{c}}} n} (\sqrt{c})^{l}\cdot (1-\sqrt{c})\\
&=1-\sqrt{c}^{p\cdot\log_{\frac{1}{\sqrt{c}}} n+1}
>1 - \frac{1}{n^p}
 \end{align}
As a result, the probability of a random walk longer than $p\cdot \log_{\frac{1}{\sqrt{c}}} n$ is smaller than $\frac{1}{n^p}$. Moreover, the upper bound in Equation \ref{equ:8} can be $\frac{3}{n^p}$ and our error in Theorem \ref{thm:error} becomes $\epsilon$ ($\epsilon \ge \frac{3}{n^p}$, $p\ge 1$).
 
 \vspace{1mm}
 \noindent
{\bf Reduce the failure probability.}
 The failure probability is derived from Equation \ref{equ:failure_pro}. We can simply reduce the failure probability by sampling more random walks. In particular, with the error threshold $\frac{3}{n^p}$, by setting $z = {(\log 2n^q)}/{2{\left(\epsilon-\frac{3}{n^p}\right)}^2}$, where $q$ is a constant and $q\ge 1$.  We can transform Equation \ref{equ:failure_pro} as follows:
 \begin{align}
    &{\bm \Pr}\left[|\tilde{s}(s,u)-s(s,u)| \ge \epsilon\right] \le 2e^{-2z{(\epsilon-\frac{3}{n^p})}^2}=\frac{1}{n^q}
\end{align}

In summary, by setting the maximum length to $p\cdot \log_{\frac{1}{\sqrt{c}}} n$ and sampling $z = {(\log 2n^q)}/{2{\left(\epsilon-\frac{3}{n^p}\right)}^2}$, we guarantee $\epsilon$ ($\epsilon \ge \frac{3}{n^p}$) error with $1-\frac{1}{n^q}$ probability, where $p,q\ge 1$. Clearly, this transformation only increase the number of random walks and walk length by a constant factor, which
does not impact the round complexity and space complexity. For example, increasing the number of samples to ${(\log 2n^2)}/{{2\left(\epsilon-\frac{3}{n^2}\right)}^2}$ and extending the maximum length to $2\log_{\frac{1}{\sqrt{c}}} n$, we can guarantee $\epsilon$ ($\epsilon \ge \frac{3}{n^2}$) error with $1-\frac{1}{n^2}$ probability.


   



\end{document}